\documentclass{article}
\usepackage{amsmath, amssymb, amsthm}
\usepackage{booktabs}

\numberwithin{equation}{section}

\newtheorem{theorem}{Theorem}

\usepackage{graphicx}
\usepackage{tablefootnote}
\begin{document}
\title{Gaussian-derivative expansion of detector resolution functions via convolution}

\author{Zan Ren\thanks{University of Chinese Academy of Sciences, Beijing, China;
\texttt{zan.ren@cern.ch}}}

\date{\today}

\maketitle

\begin{abstract}
In detector-based spectroscopy and high-energy physics, the observed lineshape of a narrow resonance is distorted by the finite resolution of the detector. While the Gaussian is the standard first approximation, realistic resolution functions deviate from it in nontrivial ways. We present a unified Fourier-space framework in which a resolution function that can be written as a Gaussian convolved with a kernel, $f=G*h$, is expanded into a series of Gaussian derivatives. The Voigtian distribution serves as the prototype, and the mechanism extends to arbitrary kernels whose Fourier transforms are analytic at the origin and of exponential type. We establish an admissibility criterion, and then examine several kernels used in practice--- smeared Hypatia, Cruijff, the Laplacian distribution, and the generalised hyperbolic core---classifying each as directly, perturbatively, or asymptotically admissible. The framework thus provides a common language for comparing and constructing non-Gaussian resolution parametrisations around the Gaussian.
\end{abstract}

\section{Introduction and motivation}

In high-energy physics (HEP) and related spectroscopic experiments, the measured invariant-mass (or energy) distribution of a narrow resonance is never the ideal line shape dictated by the underlying physics. Rather, it is distorted by the finite resolution of the detector. The standard first approximation is to model this distortion as a convolution of the intrinsic lineshape with a Gaussian $G(x;\sigma)$, motivated by the central-limit theorem: many small, independent sources of measurement uncertainty (hit-position uncertainty, tracking extrapolation, multiple scattering, calorimeter energy deposition fluctuations, etc.) add up to an approximately Gaussian smearing.

However, the Gaussian is at best an approximation. In realistic detectors, the resolution function is not perfectly Gaussian, and the deviations matter. There are two broad classes of non-Gaussian effects that one wishes to capture:

\begin{itemize}
    \item \textbf{Tail-type modifications (asymmetric power-law tails).} Radiative processes (bremsstrahlung photons) and other final-state radiation produce events that are shifted away from the nominal mass, giving rise to asymmetric tails. The Crystal Ball (CB) function and its generalisation, the Hypatia function, are prominent examples of this class: they augment a Gaussian-like core with one or two power-law tails to describe such radiative losses.
    \item \textbf{Kernel-type modifications (modified core shape).} The very \textit{core} of the resolution function may differ from a pure Gaussian, for instance because the per-event mass uncertainty itself fluctuates event-by-event, or because different detector subsystems contribute with different effective resolutions. Models such as the (smeared) Hypatia distribution, which replaces the Gaussian core by a generalised hyperbolic core, and the Cruijff function, which uses different left/right widths together with non-Gaussian tail terms, belong to this class.
\end{itemize}

This study focuses on the second class: {\it how to modify the Gaussian kernel itself}, rather than merely appending tails. The simplest conceptual idea is to regard the true resolution function as a \textit{Gaussian plus higher-order correction terms}, in the same spirit as a Taylor expansion generates higher-order terms. Convolution provides a natural and rigorous language for this idea. The Voigtian---the convolution of a Gaussian with a (non-relativistic) Breit--Wigner/Lorentzian---is the simplest non-trivial example and serves as our prototype and guiding example.

The aims of this article are therefore: (i) to show that the Voigtian admits a rigorous interpretation as a perturbative, infinite-order correction to the Gaussian; (ii) to abstract this mechanism into a general convergence framework valid for an arbitrary admissible kernel; (iii) to examine, for several concrete functions used in practice (the Laplace distribution, the generalised hyperbolic / Hypatia family, and the Cruijff function), the extent to which each can be written as a Gaussian-derivative series; and (iv) to classify them into convergently and asymptotically admissible cases.

\section{The Voigtian distribution} \label{sec:voig}
\subsection{The role of the Voigtian in modeling the detector resolution}

The Voigtian probability density function (PDF) is defined as the convolution of a Gaussian $G(x;\sigma)$ and a non-relativistic BW distribution (also known as a Lorentzian) $L(x;\gamma)$:
\begin{equation}
V(x;\sigma,\gamma) = \int_{-\infty}^{\infty} G(x';\sigma)\,L(x-x';\gamma)\,dx',
\end{equation}
where
\begin{equation}
G(x;\sigma)=\frac{1}{\sigma\sqrt{2\pi}}e^{-x^2/(2\sigma^2)},\qquad
L(x;\gamma)=\frac{\gamma/\pi}{x^2+\gamma^2}.
\end{equation}

To show that the Voigtian can be interpreted as a high-order correction to the Gaussian, we use Fourier transforms $\mathcal{F}:f(x)\to\mathcal{F}[f](k)$ to move from $x$-space to $k$-space. Denote
\begin{equation}
\mathcal{F}[f](k)=\int_{-\infty}^{\infty}f(x)e^{-ikx}\,dx.
\end{equation}
Then
\begin{equation}
\mathcal{F}[G](k)=e^{-\sigma^2 k^2/2},\qquad
\mathcal{F}[L](k)=e^{-\gamma|k|}.
\end{equation}
By the convolution theorem,
\begin{equation}
\mathcal{F}[V](k)=e^{-\sigma^2 k^2/2}\, e^{-\gamma|k|}.
\end{equation}

Thus
\begin{equation}
V(x)=\frac{1}{2\pi}\int_{-\infty}^{\infty} e^{-ikx}\, e^{-\sigma^2 k^2/2}\, e^{-\gamma|k|}\,dk.
\end{equation}

Expand $e^{-\gamma|k|}$ in a Maclaurin series in the small parameter $\gamma$:
\begin{equation}
e^{-\gamma|k|}=\sum_{n=0}^{\infty}\frac{(-\gamma)^n}{n!}\,|k|^n.
\end{equation}

Interchanging sum and integral\footnote{This is justified by uniform convergence on any bounded $k$-interval, together with the suppressing effect of the Gaussian factor.} gives
\begin{equation}
V(x)=\sum_{n=0}^{\infty}\frac{(-\gamma)^n}{n!}\,I_n(x),
\qquad
I_n(x)=\frac{1}{2\pi}\int_{-\infty}^{\infty} e^{-ikx}\, e^{-\sigma^2 k^2/2}\,|k|^n\,dk.
\end{equation}

For $n=0$, $I_0(x)=G(x;\sigma)$, i.e.\ {\bf the zeroth-order term is the Gaussian itself}. For $n\ge 1$, the integrals $I_n(x)$ are related to derivatives of the Gaussian or of its Hilbert transform:

\begin{itemize}
    \item If $n=2m$ (even), $|k|^{2m}=k^{2m}$, and
    \begin{equation}
    I_{2m}(x)=(-1)^m\,\frac{d^{2m}}{dx^{2m}}G(x).
    \end{equation}
    \item If $n=2m+1$ (odd), $|k|^{2m+1}=k^{2m+1}\,\operatorname{sgn}(k)$. Its inverse Fourier transform involves the Hilbert transform $H[G](x)=\frac{1}{\pi}\,\mathrm{P.V.}\!\int\frac{G(x')}{x-x'}\,dx'$:
    \begin{equation}
    I_{2m+1}(x)=(-1)^m\,\frac{d^{2m+1}}{dx^{2m+1}}H[G](x).
    \end{equation}
\end{itemize}

Therefore, the Voigtian distribution can be written as
\begin{equation}
V(x)=G(x)+\sum_{m=1}^{\infty}\frac{(-\gamma)^{2m}}{(2m)!}(-1)^m G^{(2m)}(x)
+\sum_{m=0}^{\infty}\frac{(-\gamma)^{2m+1}}{(2m+1)!}(-1)^m H[G]^{(2m+1)}(x).
\end{equation}

The first term is the original Gaussian $G(x)$. The second series contains even-order derivatives of $G(x)$, weighted by powers $\gamma^2,\gamma^4,\dots$. The third series contains odd-order derivatives of the Hilbert transform of $G$, weighted by $\gamma^1,\gamma^3,\dots$. Because the Voigtian is an even function, the odd-order terms combine symmetrically.

In the limit $\gamma\to0$, all higher-order correction terms disappear and $V(x)\to G(x)$. For small but finite $\gamma$, the Voigtian represents a perturbation expansion around the ideal Gaussian, where each derivative term modifies the shape---altering the peak height, width, and tail behaviour (e.g.\ introducing power-law tails, of the kind also produced by CB or Hypatia-type models).

\subsection{Why the Lorentz form ensures convergence}

The key lies in the Fourier-domain property of the Lorentzian,
\begin{equation}
\mathcal{F}[L](k)=e^{-\gamma|k|}.
\end{equation}
This function has two critical features:
\begin{itemize}
    \item It is {\bf non-analytic at $k=0$} (the absolute value creates a cusp), yet when multiplied by the Gaussian factor $e^{-\sigma^2 k^2/2}$ the product becomes $C^\infty$ at $k=0$, because the Gaussian smooths out the singularity.
    \item As $|k|\to\infty$, $|\mathcal{F}[V](k)|\lesssim e^{-\sigma^2 k^2/2}$, which decays {\bf super-exponentially}. This guarantees that the inverse Fourier transform $V(x)$ is a real-analytic function, so its derivative expansion converges everywhere.
\end{itemize}
Thus the Lorentzian provides exactly the right balance: a non-trivial tail structure in coordinate space while remaining sufficiently well behaved in Fourier space to admit a convergent expansion.

\subsection{Numerical validation}

The derivative-expansion mechanism can be illustrated numerically for the prototypical Voigtian case. Figure~\ref{fig:voigt_check1} compares the exact Voigtian with the FFT-based numerical convolution $G*L$ and with truncated derivative-series approximations. The FFT convolution agrees with the analytic Voigt to within a few percent, validating the numerical implementation. The derivative-series truncation is asymptotic in nature (the cusp of $e^{-\gamma|k|}$ at $k=0$ limits the convergence radius), but its quality improves systematically as $\gamma/\sigma$ decreases.

\begin{figure}[htbp]
\centering
\includegraphics[width=\textwidth]{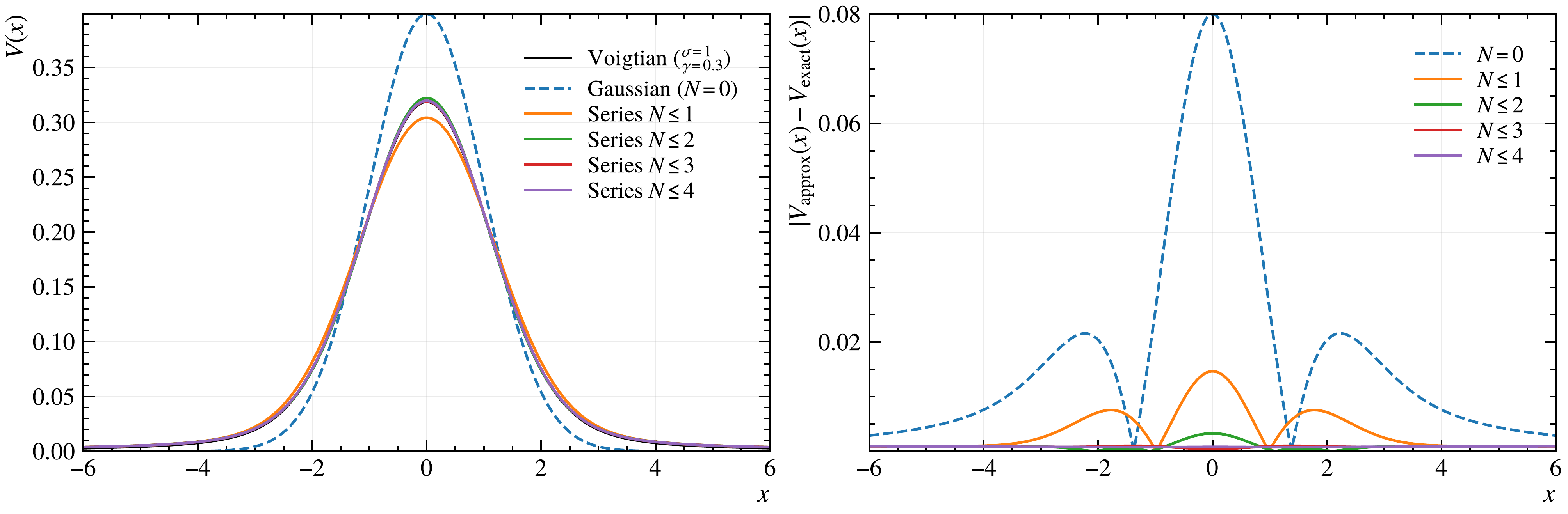}
\caption{Profile comparison (left) and relative-error plot (right) for the Voigtian at $\sigma=1$, $\gamma=0.3$.}
\label{fig:voigt_check1}
\end{figure}

\section{General conditions and admissibility criterion}

We now formulate the precise question addressed in this work. Given a resolution function $f(x)$, we ask whether it can be represented as a series of Gaussian derivatives,
\begin{equation}
f(x)=\sum_{n=0}^{\infty} c_n\,G^{(n)}(x),
\label{eq:target_form}
\end{equation}
for some coefficients $c_n$. By the Fourier identity $\mathcal{F}[G^{(n)}](k)=(-ik)^n e^{-\sigma^2 k^2/2}$, this is equivalent to writing the Fourier transform of $f$ as
\begin{equation}
\hat{f}(k)=e^{-\sigma^2 k^2/2}\,\hat{h}(k),\qquad \hat{h}(k):=\sum_{n=0}^{\infty} c_n\,(ik)^n,
\label{eq:fourier_target}
\end{equation}
where $\hat{h}$ is the Fourier transform of some kernel $h$ with $f=G*h$. Thus \emph{$f$ admits a Gaussian-derivative expansion if and only if $f=G*h$ for a kernel $h$ whose Fourier transform $\hat{h}$ is entire (so that the Taylor series $\sum c_n(ik)^n$ converges) and whose growth is compatible with the Gaussian damping}. The precise sufficient condition is given by the following theorem.

\begin{theorem}[Admissibility criterion]
\label{thm:admissible}
Let $f\in L^1(\mathbb{R})$ be a resolution function with Fourier transform $\hat{f}$, and write $\hat{f}(k)=e^{-\sigma^2 k^2/2}\,\hat{h}(k)$. Suppose that
\begin{enumerate}
    \item $\hat{h}$ is analytic in a neighbourhood of $k=0$ and admits a Taylor expansion $\hat{h}(k)=\sum_{n\ge0}\hat{h}^{(n)}(0)k^n/n!$ valid for $|k|<\delta$;
    \item there exist constants $C,s>0$ such that $|\hat{h}(k)|\le C e^{s|k|}$ for all $k\in\mathbb{R}$ (exponential-type bound).
\end{enumerate}
Then $f=G*h$ admits the expansion
\begin{equation}
f(x)=\sum_{n=0}^{\infty}\frac{\hat{h}^{(n)}(0)}{n!}\,(-i)^n\,G^{(n)}(x),
\label{eq:general_series}
\end{equation}
which converges absolutely and uniformly on compact $x$-sets.
\end{theorem}

\begin{proof}
Under the change of variable $\hat{h}(k)=\hat{f}(k)e^{+\sigma^2 k^2/2}$, the representation \eqref{eq:fourier_target} is exact. Condition~(i) gives a convergent Taylor expansion of $\hat{h}$ near $k=0$; condition~(ii) controls its growth.\footnote{The criterion is stated in terms of $\hat{h}=\hat{f}\,e^{+\sigma^2 k^2/2}$ rather than $\hat{f}$ alone, because the Gaussian factor $e^{-\sigma^2 k^2/2}$ in $\hat{f}$ is already accounted for by the Gaussian-derivative basis. In practice, to test whether a given $f$ is admissible, one computes $\hat{h}(k)=\hat{f}(k)e^{+\sigma^2 k^2/2}$ and checks conditions~(i)--(ii).} Splitting the inverse-Fourier integral for $f$ at $|k|=\delta$, the inner part is expanded termwise (uniform convergence on $|k|\le\delta$ justifies interchange), while the outer part is bounded by $C e^{-\sigma^2 k^2/2+s|k|}$, which is integrable because the Gaussian quadratic dominates any linear exponent for large $|k|$. The remainder after $N$ terms is controlled by $\sum_{n>N} s^n/n!\to0$. Hence the series \eqref{eq:general_series} converges absolutely and uniformly on compact sets.
\end{proof}

\section{Extensions to specific kernels} \label{sec:extension}

We now examine several resolution functions used in HEP and statistics. For each $f$, we form $\hat{h}(k)=\hat{f}(k)e^{+\sigma^2 k^2/2}$ and test whether $\hat{h}$ satisfies the two conditions of Theorem~\ref{thm:admissible}. If both hold, $f$ is \emph{convergently admissible} ($f=G*h$ exactly, with convergent series). If $\hat{h}$ is analytic at $k=0$ but fails the exponential bound (so the formal series diverges), $f$ is \emph{asymptotically admissible} (a controlled asymptotic expansion for parameters near the Gaussian limit). If $f$ is not of the form $G*h$ at all, we indicate the perturbative interpretation where available.

\subsection{Smeared Hypatia: convergently admissible} \label{sec:hypatia}

The Hypatia distribution was introduced by Andreetto \textit{et al.}~\cite{Andreetto:2013} as a generalisation of the Crystal Ball function. Its \emph{smeared} form is explicitly defined as a convolution:
\begin{equation}
\Upsilon(m;\mu,\sigma^{SR},\lambda,\zeta,\beta,a,n,v_0)
= \frac{1}{\sqrt{v_0}}\,e^{-\frac{1}{2v_0}m^2}\;*\;I(m;\mu,\sigma^{SR},\lambda,\zeta,\beta,a,n),
\end{equation}
where $I$ is the unsmeared Hypatia (generalised-hyperbolic) core. This is exactly $f=G*h$ with $h=I$. The kernel $h$ is a generalised hyperbolic density; its Fourier transform $\hat{h}(k)$ is analytic in a neighbourhood of $k=0$ and satisfies an exponential bound $|\hat{h}(k)|\le C e^{s|k|}$ (the same property used for the Apollonios core below, Eq.~\ref{eq:GH_charfunc}). Therefore Theorem~\ref{thm:admissible} applies directly, and the smeared Hypatia function admits the convergent expansion
\begin{equation}
\Upsilon(x) = \sum_{n=0}^{\infty} \frac{\hat{I}^{(n)}(0)}{n!}\,(-i)^n G^{(n)}(x),
\label{eq:hypatia_expansion}
\end{equation}
with absolute convergence uniform on compact sets. The coefficients $\hat{I}^{(n)}(0)$ are the Taylor coefficients of the Hypatia core's characteristic function at $k=0$; they can be obtained by differentiating the closed-form generalised-hyperbolic characteristic function (see Eq.~\ref{eq:GH_charfunc}) at $k=0$, either symbolically or numerically. Thus Hypatia is a prime example of a kernel-type modification that fits naturally into our framework with guaranteed convergence. For the unsmeared Hypatia PDF, see Section~\ref{sec:apollonis}.

\subsection{The Cruijff function: perturbatively admissible}

The Cruijff function is defined as
\begin{equation}
f_{\text{C}}(x;\mu,\sigma_L,\sigma_R,\alpha_L,\alpha_R)
= \exp\!\left(-\frac{(x-\mu)^2}{2\sigma_{L,R}^2 + \alpha_{L,R}(x-\mu)^2}\right),
\label{eq:cruijff}
\end{equation}
with $\sigma_L,\alpha_L$ for $x<\mu$ and $\sigma_R,\alpha_R$ for $x\ge\mu$. This is \emph{not} of the form $G*h$ for any kernel $h$; it is a direct modification of the Gaussian exponent rather than a convolution. Therefore the rigorous convergence theorem does not apply.

Nevertheless, for small $\alpha_{L,R}$, the denominator can be expanded:
\begin{equation}
\frac{1}{2\sigma^2 + \alpha (x-\mu)^2} = \frac{1}{2\sigma^2} \sum_{m=0}^{\infty} \left(-\frac{\alpha}{2\sigma^2}\right)^m (x-\mu)^{2m},
\label{eq:cruijff_expansion}
\end{equation}
leading to a series of polynomial corrections to the Gaussian exponent. In Fourier space, these corrections generate higher-order derivative terms structurally analogous to those in the Gaussian-derivative expansion. Hence Cruijff can be interpreted \emph{perturbatively} as a Gaussian kernel dressed with higher-order corrections, although it does not satisfy the convolution structure required for the exact derivative series.

\subsection{The Laplace distribution: asymptotically admissible}

Consider the Laplace (double-exponential) distribution,
\begin{equation}
f_L(x;\lambda)=\frac{1}{2\lambda}e^{-|x|/\lambda},\qquad \lambda>0,
\label{eq:laplace_pdf}
\end{equation}
as the resolution function to be expanded. Its Fourier transform is
\begin{equation}
\hat{f}_L(k)=\frac{1}{1+\lambda^2 k^2}.
\label{eq:laplace_fourier}
\end{equation}
To test admissibility, form the candidate kernel's Fourier transform:
\begin{equation}
\hat{h}(k) = \hat{f}_L(k)\,e^{+\sigma^2 k^2/2} = \frac{e^{\sigma^2 k^2/2}}{1+\lambda^2 k^2}.
\label{eq:laplace_deconv}
\end{equation}
As $|k|\to\infty$, $\hat{h}(k)$ grows super-exponentially and is not the Fourier transform of any tempered distribution. Hence $f_L$ is \emph{not} of the form $G*h$ in the usual sense, and Theorem~\ref{thm:admissible} does not apply.

Nevertheless, $\hat{f}_L(k)$ is analytic at $k=0$ and admits the Taylor expansion
\begin{equation}
\hat{f}_L(k)=\sum_{m=0}^{\infty}(-\lambda^2 k^2)^m,\qquad |k|<\frac{1}{\lambda},
\label{eq:laplace_taylor}
\end{equation}
so that $\hat{h}(k)$ is also analytic at $k=0$ with Taylor coefficients determined by the product of the two series. Substituting $\hat{f}_L(k)=\sum_{m\ge0}(-\lambda^2 k^2)^m$ into the inverse Fourier representation of $f_L$ and interchanging sum and integral gives the \emph{formal} Gaussian-derivative series
\begin{equation}
f_L(x) \stackrel{\text{formal}}{=} \sum_{m=0}^{\infty} \lambda^{2m}\,G^{(2m)}(x).
\label{eq:laplace_formal}
\end{equation}
This series does \emph{not} converge for any $x$; it is an asymptotic expansion. Indeed, the $m$-th term behaves asymptotically like
\begin{equation}
\lambda^{2m}\,G^{(2m)}(x) \sim \lambda^{2m}\,\frac{(2m)!}{\sigma^{2m}\,m!}\qquad (m\to\infty),
\label{eq:laplace_asymp}
\end{equation}
using the large-order asymptotics of the Hermite-polynomial coefficients of $G^{(2m)}$; for fixed $\lambda>0$ this grows super-exponentially, so the series diverges.

However, for $\lambda\ll\sigma$ (the Laplace distribution much narrower than the Gaussian scale), the first few terms give an excellent approximation: truncating at $m=N$ leaves a remainder of order $\mathcal{O}(\lambda^{2N+2})$, making the expansion practically useful in the near-Gaussian regime. Thus the Laplace distribution is \emph{asymptotically admissible}: it admits a formal Gaussian-derivative series that is divergent but controlled for small $\lambda/\sigma$.

\subsection{The generalised hyperbolic core (Apollonios): asymptotically admissible} \label{sec:apollonis}

The \textit{Apollonios} distribution is defined here as a member of the generalised hyperbolic (GH) family of probability densities~\cite{Andreetto:2013}. Its unsmeared core $I_A(x)$ is a GH density with parameters $(\lambda,\alpha,\beta,\delta,\mu)$. For completeness we record the density and its known closed-form characteristic function; these are standard identities for the GH family and supply the Fourier data needed below.

The GH density is
\begin{align}
I_A(x;\lambda,\alpha,\beta,\delta,\mu)
=& \frac{(\alpha^{2}-\beta^{2})^{\lambda/2}}{\sqrt{2\pi}\,\alpha^{\lambda-1/2}\,
  \delta^{\lambda}\,K_{\lambda}\!\big(\delta\sqrt{\alpha^{2}-\beta^{2}}\big)} \\ \nonumber
  \cdot & \Bigl(\delta^{2}+(x-\mu)^{2}\Bigr)^{\lambda/2-1/4}
  K_{\lambda-1/2}\!\Bigl(\alpha\sqrt{\delta^{2}+(x-\mu)^{2}}\Bigr)\,
  e^{\beta(x-\mu)},
\label{eq:GH_density}
\end{align}
where $K_\nu(\cdot)$ is the modified Bessel function of the second kind, $\alpha>|\beta|\ge0$, $\delta>0$~\cite{Prause1999,Paolella2007}. The GH family contains as special cases the Student-$t$, Laplace, hyperbolic, normal-inverse-Gaussian (NIG) and variance-gamma distributions.

Its characteristic function is known in closed form:
\begin{equation}
\hat{I}_A(k)=
\frac{e^{i\mu k}\,
      (\alpha^{2}-\beta^{2})^{\lambda/2}}
     {\bigl(\alpha^{2}-(\beta+ik)^{2}\bigr)^{\lambda/2}}
\frac{K_{\lambda}\!\bigl(\delta\sqrt{\alpha^{2}-(\beta+ik)^{2}}\bigr)}
     {K_{\lambda}\!\bigl(\delta\sqrt{\alpha^{2}-\beta^{2}}\bigr)},
\qquad \alpha>|\beta|.
\label{eq:GH_charfunc}
\end{equation}
Two properties of $\hat{I}_A(k)$ are relevant. First, analyticity near $k=0$: for $\alpha>|\beta|$ the argument $\alpha^{2}-(\beta+ik)^{2}$ stays away from the branch cut of $K_\lambda$ in the strip $|\operatorname{Im}k|<\alpha-|\beta|$, so $\hat{I}_A(k)$ is analytic there and in particular near $k=0$; hence it admits a convergent Taylor expansion $\hat{I}_A(k)=\sum_{n\ge0}\hat{I}_A^{(n)}(0)k^n/n!$. Second, exponential-type bound: as $|k|\to\infty$ along the real axis, $\hat{I}_A(k)$ decays like $e^{-\delta|k|}$, i.e.\ $|\hat{I}_A(k)|\le C e^{\delta|k|}$.

Now treat $I_A$ as the resolution function $f$ to be expanded, and form
\begin{equation}
\hat{h}(k)=\hat{I}_A(k)\,e^{+\sigma^2 k^2/2}.
\label{eq:gh_deconv}
\end{equation}
Because $\hat{I}_A(k)$ decays only like $e^{-\delta|k|}$ while $e^{+\sigma^2 k^2/2}$ grows super-exponentially, $\hat{h}(k)$ diverges as $|k|\to\infty$ and is not the Fourier transform of any tempered distribution. Therefore $I_A$ is \emph{not} of the form $G*h$, and Theorem~\ref{thm:admissible} does not apply.

Nevertheless, since $\hat{I}_A(k)$ is analytic at $k=0$, substituting its Taylor expansion into the inverse Fourier representation of $I_A$ yields the \emph{formal} Gaussian-derivative series
\begin{equation}
I_A(x) \stackrel{\text{formal}}{=} \sum_{n=0}^{\infty} \frac{\hat{I}_A^{(n)}(0)}{n!}\,(-i)^n\,G^{(n)}(x).
\label{eq:gh_formal}
\end{equation}
This series is asymptotic, not convergent: the Taylor coefficients $\hat{I}_A^{(n)}(0)$ grow too fast to be compensated by the Gaussian derivatives. We now show, however, that for large $\delta$ the asymptotic series is well behaved and low-order truncations are accurate---the basis for calling the GH core \emph{asymptotically admissible}.

We set $\mu=0$ for simplicity. The key is the local behaviour of $\hat{I}_A(k)$ near $k=0$. Define
\begin{equation}
\phi(k):=\sqrt{\alpha^{2}-(\beta+ik)^{2}}-\sqrt{\alpha^{2}-\beta^{2}},
\qquad \phi(0)=0.
\label{eq:GH_phi}
\end{equation}
Using the large-argument expansion $K_\nu(z)\sim\sqrt{\pi/(2z)}\,e^{-z}\bigl(1+\mathcal{O}(1/z)\bigr)$ as $z\to\infty$ in \eqref{eq:GH_charfunc}, the exponential factors in numerator and denominator cancel, and the leading dependence on $k$ is
\begin{equation}
\hat{I}_A(k)\sim \exp\!\bigl(-\delta\,\phi(k)\bigr)\qquad (\delta\to\infty,\ |k|\ \text{fixed}).
\label{eq:GH_local_gauss}
\end{equation}
Expanding $\phi(k)=\dfrac{i\beta}{\sqrt{\alpha^{2}-\beta^{2}}}\,k-\dfrac{1}{2\sqrt{\alpha^{2}-\beta^{2}}}\,k^{2}+\mathcal{O}(k^{3})$, the exponent becomes
$-\delta\phi(k)=i\mu_\delta k-\tfrac12\sigma_\delta^{2}k^{2}+\mathcal{O}(k^{3})$ with $\mu_\delta=-\delta\beta/\sqrt{\alpha^{2}-\beta^{2}}$ and $\sigma_\delta^{2}=\delta/\sqrt{\alpha^{2}-\beta^{2}}$. Thus, to leading order, $\hat{I}_A(k)\approx e^{i\mu_\delta k-\frac12\sigma_\delta^{2}k^{2}}$---the characteristic function of a Gaussian of width $\sigma_\delta\propto\sqrt{\delta}$. This is the precise meaning of ``the GH core approaches a Gaussian for large $\delta$''.

Because $\hat{I}_A$ is analytic at $k=0$, Cauchy's estimate gives, for any $R$ inside the strip of analyticity,
\begin{equation}
|\hat{I}_A^{(n)}(0)|\le \frac{n!}{R^{\,n}}\,\max_{|k|=R}|\hat{I}_A(k)|.
\label{eq:GH_cauchy}
\end{equation}
Choosing $R=\epsilon/\sqrt{\delta}$ ($\epsilon>0$ fixed) and using $|\hat{I}_A(k)|\lesssim e^{-\delta\operatorname{Re}\phi(k)}\approx e^{-\epsilon^{2}/2}$ on $|k|=R$, we obtain
$|\hat{I}_A^{(n)}(0)|\lesssim n!\,(\sqrt{\delta}/\epsilon)^{n}e^{-\epsilon^{2}/2}$. Hence the $n$-th term of the formal series satisfies
\begin{equation}
\left|\frac{\hat{I}_A^{(n)}(0)}{n!}\,(-i)^{n}G^{(n)}(x)\right|
\lesssim \left(\frac{\sqrt{\delta}}{\epsilon\,\sigma}\right)^{\!n}\!|G^{(n)}(x)|,
\label{eq:GH_term}
\end{equation}
where $|G^{(n)}(x)|$ is of Hermite-polynomial type and decays rapidly once $n$ exceeds $\sigma^{2}$. For fixed truncation order $N$, the remainder after $N$ terms behaves as
\begin{equation}
R_N(x):=\left|I_A(x)-\sum_{n=0}^{N}\frac{\hat{I}_A^{(n)}(0)}{n!}\,(-i)^{n}G^{(n)}(x)\right|
=\mathcal{O}\!\left(\delta^{-(N+1)/2}\right),
\label{eq:GH_remainder}
\end{equation}
so the truncation error decreases as $\delta$ grows and as higher orders are retained. The series is therefore a \emph{controlled asymptotic expansion}: divergent in the strict sense (the factorial growth of $\hat{I}_A^{(n)}(0)$ ultimately dominates), but accurate when truncated at low orders for $\delta$ large enough that the GH core is close to Gaussian. This is the same mechanism as for the Laplace distribution, with $\delta^{-1/2}$ playing the role of the small parameter $\lambda/\sigma$. Thus the GH core is \emph{asymptotically admissible}.

In experimental applications such as B-meson mass fits at LHCb or Belle, the Hypatia/GH core is used to marginalise over per-event mass uncertainty, with $\delta$ determined from fits to simulated samples. The fitted $\delta$ takes moderate values chosen so that the resulting core width $\sigma_f\approx\sqrt{\delta/\sqrt{\alpha^{2}-\beta^{2}}}$ matches the detector resolution — typically $\sigma_f\sim 5$–$15\ \mathrm{MeV}/c^{2}$ for B decays, against a mass scale of order $5\times10^{3}\ \mathrm{MeV}/c^{2}$. Consequently $\delta$ is \emph{not} in the extreme large-$\delta$ Gaussian limit; the GH core retains non-negligible non-Gaussian features, consistent with its classification as asymptotically admissible. The convergent case is recovered only after the explicit smearing convolution $G*I$ that defines the smeared Hypatia distribution, where the Gaussian factor enters the Fourier transform of the admissible object.

\subsection{Summary of admissibility}

Table~\ref{tab:kernels} summarises the status of each function discussed.

\begin{table}[htbp]
\caption{Summary of resolution functions and their admissibility.}
\label{tab:kernels}
\centering
\begin{tabular}{p{3.2cm} p{3.5cm} p{5.5cm}}
\toprule
{\bf Function $f$} & {\bf $\hat{h}(k)=\hat{f}(k)e^{+\sigma^2 k^2/2}$} & {\bf Admissibility \& Remarks}\\
\midrule
Voigtian $V=G*L$ & $e^{-\gamma|k|}$ (exp.\ decay) & Convergent series (prototype)\\
\hline
Smeared Hypatia $\Upsilon=G*I$ & $\hat{I}$, gen.\ hyperbolic, exp.-bounded & Convergent series\\
\hline
Cruijff & not $G*h$ & Perturbative (small-$\alpha$ expansion)\\
\hline
Laplace $f_L$ & $\frac{e^{\sigma^2 k^2/2}}{1+\lambda^2 k^2}$, super-exp.\ growth & Asymptotic series; useful for $\lambda\ll\sigma$\\
\hline
GH core $I_A$ & $\hat{I}_A(k)e^{+\sigma^2 k^2/2}$, super-exp.\ growth & Asymptotic series; useful for large $\delta$\\
\bottomrule
\end{tabular}
\end{table}

\section{Summary and conclusions}

We have presented a unified Fourier-space framework for representing detector resolution functions as series of Gaussian derivatives. The central question is whether a given $f$ can be written as $f=\sum_n c_n G^{(n)}$, equivalently $f=G*h$ with $\hat{h}=\hat{f}\,e^{+\sigma^2 k^2/2}$ satisfying appropriate analyticity and growth conditions. The Voigtian serves as the prototype: expanding its Fourier transform $e^{-\gamma|k|}$ generates a derivative series in which the Gaussian is the zeroth-order term and higher orders are Gaussian derivatives (even) and Hilbert-transform derivatives (odd). Convergence is guaranteed by the super-exponential damping of the Gaussian factor, and Theorem~\ref{thm:admissible} extends this to any $f=G*h$ whose kernel $h$ has an exponentially bounded Fourier transform.

We have shown that the framework naturally distinguishes three levels of admissibility. The Voigtian and the smeared Hypatia distribution are \emph{convergently admissible}: they are exactly of the form $f=G*h$ and their Gaussian-derivative series converge. The Cruijff function, though not a convolution, admits a perturbative interpretation for small tail parameters. Finally, the Laplace distribution and the generalised hyperbolic core (Apollonios) are \emph{asymptotically admissible}: they can be formally expanded as Gaussian-derivative series, but the series diverge. Nevertheless, for parameters close to the Gaussian limit ($\lambda\ll\sigma$ or $\delta\gg1$), the truncated series provide controlled approximations. This hierarchy clarifies the domain of applicability of the derivative-expansion approach and unifies seemingly disparate resolution models under a common analytical framework.


\begin{thebibliography}{9}
\bibitem{Andreetto:2013}
P.~Andreetto {\em et al.},
{\em Mass distributions marginalized over per-event errors},
Nucl.\ Instrum.\ Meth.\ A {\bf 748} (2014) 113,
arXiv:1312.5000 [physics.data-an].

\bibitem{LHCb:2018pqq}
LHCb Collaboration,
{\em Evidence for an $\eta_c(1S)\pi^-$ resonance in $B^0\to\eta_c(1S)K^+\pi^-$ decays},
Eur.\ Phys.\ J.\ C {\bf 78} (2018) 1019,
arXiv:1809.07416 [hep-ex].

\bibitem{LHCb:2017pqq}
LHCb Collaboration,
{\em Observation of the decays $\Lambda_b^0\to\chi_{c1}pK^-$ and $\Lambda_b^0\to\chi_{c2}pK^-$},
Phys.\ Rev.\ D {\bf 96} (2017) 092005,
arXiv:1709.01980 [hep-ex].

\bibitem{CrystalBall}
M.~Oreglia,
{\em A study of the reactions $\psi'\to\gamma\gamma\psi$} (Ph.D.\ thesis),
SLAC Report SLAC-236 (1980);\\
T.~Skwarnicki,
{\em A study of the radiative cascade transitions between the Upsilon-prime and Upsilon resonances} (Ph.D.\ thesis),
DESY F31-86-02 (1986).

\bibitem{Prause1999}
K.~Prause,
{\em The generalized hyperbolic model: estimation, financial derivatives
 and risk measures},
PhD thesis, University of Freiburg (1999).

\bibitem{Paolella2007}
M.~S.~Paolella,
{\em Intermediate Probability: A Computational Approach},
Wiley (2007), Chapter 8.
\end{thebibliography}
\end{document}